\documentclass{article}
\usepackage[margin=1in]{geometry}
\usepackage{graphicx}
\usepackage{color}

\usepackage{setspace}
\usepackage{enumerate}
\usepackage{verbatim}
\usepackage{amssymb}
\usepackage{amsbsy}
\usepackage{theorem}
\usepackage{tikz}
\usepackage{url}

\usepackage[dvips]{epsfig}
\usepackage{pgfplots}
\usepackage{amsmath}

\usepackage{mathtools}

\newcommand*\diff{\mathop{}\!\mathrm{d}}

\def\urltilde{\kern -.15em\lower .7ex\hbox{\~{}}\kern .04em}
\def\urldot{\kern -.10em.\kern -.10em}

\def\urlhttp{http\kern -.10em\lower -.1ex\hbox{:}\kern -.12em\lower 0ex\hbox{/}\kern -.18em\lower 0ex\hbox{/}}

\usetikzlibrary{decorations.markings}
\tikzset{
  set arrow inside/.code={\pgfqkeys{/tikz/arrow inside}{#1}},
  set arrow inside={end/.initial=stealth, opt/.initial=},
  /pgf/decoration/Mark/.style={
    mark/.expanded=at position #1 with
    {
      \noexpand\arrow[\pgfkeysvalueof{/tikz/arrow inside/opt}]{\pgfkeysvalueof{/tikz/arrow inside/end}}
    }
  },
  arrow inside/.style 2 args={
    set arrow inside={#1},
    postaction={
      decorate,decoration={
        markings,Mark/.list={#2}
      }
    }
  },
}

\theoremstyle{plain}
\newtheorem{Theorem}{Theorem}

\newtheorem{Corollary}{Corollary}

{\theorembodyfont{\rmfamily} \newtheorem{Example}{Example}}
\newtheorem{Remark}{Remark}

\newcommand{\qedwhite}{\hfill \ensuremath{\Box}}

\newenvironment{proof}[1][Proof]{\noindent\textbf{#1.} }{\ \rule{0.5em}{0.5em}}

\newcommand {\R}{\mathbb R}

\newcommand{\be}{\begin{equation}}
\newcommand{\ee}{\end{equation}}

\newcommand{\Adj}{\operatorname{{\mathrm adj}}}

{\theorembodyfont{\rmfamily} }

\newcommand{\diag}{\operatorname{{\mathrm diag}}}

\usepackage{color}

\begin{document}
\title{Approximating the Frequency Response of Contractive Systems}
\date{}
 \author{Michael Margaliot\thanks{M. Margaliot
 is with the School of Electrical Engineering and the Sagol School of Neuroscience, Tel-Aviv. 
 University, Tel-Aviv 69978, Israel.
 E-mail: \texttt{michaelm@eng.tau.ac.il}. The research of MM is partially supported by    research grants from  the Israeli Ministry of Science, Technology \& Space, the US-Israel Binational Science Foundation,  and   the
Israel Science Foundation~(ISF grant 410/15)} \and Samuel Coogan\thanks{S. Coogan is with the Department of  Electrical Engineering, University of California, Los Angeles. E-mail: \texttt{scoogan@ucla.edu}. }
 }

\maketitle
\begin{abstract}
We consider contractive systems 
whose trajectories evolve on a compact and convex state-space. It is well-known that if 
 the time-varying
 vector field of the system
 is periodic then 
 the system admits a unique globally asymptotically stable periodic solution.   
Obtaining explicit information on this periodic solution and its dependence
 on various parameters is important both theoretically and in 
numerous    applications. 
We develop an approach for approximating such a periodic trajectory using the periodic 
trajectory of a simpler system (e.g. an LTI system). Our approximation includes an error bound 
that is based on the input-to-state stability property of contractive systems. 
 We show that in some cases this error bound can be computed explicitly.
We also use the bound to derive a new theoretical result, namely, 
that a contractive system with an additive periodic input behaves like a low pass filter.
 We demonstrate our results
using  several  examples from  systems biology. 

 \end{abstract}
\section{Introduction}\label{sec:intro}

A dynamical system is called \emph{contractive}
 if any two trajectories approach each other~\cite{LOHMILLER1998683,sontag_cotraction_tutorial}. 
This is a strong property with   many important 
implications. For example, if the trajectories evolve on a compact and convex
state-space~$\Omega$ then the system admits an equilibrium point~$e\in \Omega$, 
and since every trajectory converges to the trajectory emanating from~$e$,
  $e$ is globally 
  asymptotically   stable. 	Note that establishing
	this does not require an explicit description of~$e$. 

More generally, contractive systems with a periodic
 excitation \emph{entrain}, that is, their trajectories converge to a periodic solution with the same period as the excitation.  
This property is very important in applications ranging from entrainment 
of biological systems to periodic excitations (e.g., the 24h solar day or the periodic cell-cycle division
 program) 
to the entrainment  of synchronous generators to the frequency of the electric grid. 
However, the proof of the entrainment property of contractive systems is based on implicit arguments
(see, e.g.~\cite{entrain2011})
and  provides no explicit  information on
 the periodic trajectory (except for its period).

Contraction theory has found numerous applications in  systems and control theory, systems biology~\cite{Russo2011_book_chap}, and more
(see e.g. the recent survey~\cite{sontag_cotraction_tutorial}). 
A particularly interesting line of research is based on combining contraction theory and graph theory
in order to study various networks of multi-agent systems (see, e.g.~\cite{hier_cont,Arcak20111219,cont_slotine_graph, coogan2015compartmental}).

As already noted by Desoer and Haneda~\cite{Desoer_cont}, 
contractive systems satisfy a special case of the input-to-state stability~(ISS)
  property (see the survey paper~\cite{sontag2008}). 
Desoer and Haneda used this to derive   
bounds on the error between trajectories of a continuous-time
 contractive system and its
time-discretized model. This is important when computing solutions of contractive systems using numerical integration methods~\cite{Arcak2015}. Sontag~\cite{con_with_input} has shown   
that contractive systems satisfy a ``converging-input converging output'' property.
A recent paper~\cite{lars16}  used the~ISS property 
to derive a bound on the  error between trajectories of a continuous-time
  contractive system and   those of some ``simpler'' continuous-time system (e.g. an LTI system). 
This  bound is particularly useful when the simpler  model can be solved explicitly.

Here, we derive new  bounds on the distance between the periodic trajectory of a 
  contractive system and the periodic trajectory of a ``simpler'' system, e.g. an LTI system with a periodic forcing. 
We show several cases where the periodic trajectory of the simpler system is explicitly known and the bound is also explicit, so this
provides considerable information on the unknown periodic trajectory of the contractive system. 
The explicit bounds also pave the way for new theoretical results. We demonstrate this by using one of the bounds to prove that any contractive system with an additive sinusoidal  forcing behaves like a low-pass filter, that is,
as the frequency of the sinusoidal signal goes to infinity the corresponding solution of the 
system converges to an equilibrium state. This generalizes the well-known
 behavior of asymptotically stable LTI systems. 

The remainder of this paper 
 is organized as follows. The next section reviews some properties of contractive systems  and in particular their~ISS property.
For more details, 
including the historic development of contraction theory, see e.g.~\cite{soderling_survey,cont_anc}.
The next three sections describe our main results. 
Section~\ref{sec:bounds} develops a bound for the difference between the periodic trajectories
 of two systems:
a contractive system and some simpler
 ``approximating'' system. We show using an example that in general this bound cannot be improved.
 Section~\ref{sec:approx} suggests two possible approximating systems for the case of a contractive system with
a periodic forcing. 
Section~\ref{sect:theory} shows how the explicit bounds can be used to derive
 a  new
theoretical result on  the frequency response of contractive systems. 
  The final section concludes and describes 
possible directions for further research.

\section{Preliminaries}

 Consider the time-varying dynamical
 system
 \be\label{eq:fdyn}
            \dot{x}(t)=f(t,x(t)),
 \ee
with the state $x$ evolving on a  positively invariant
convex set~$\Omega \subseteq \R^n $. We assume that~$f(t,x)$ is differentiable with respect to~$x$, and that
both~$f(t,x)$ and its Jacobian~$J(t,x):=\frac{\partial f}{\partial x}(t,x)$ are continuous in~$(t,x)$.
Let~$x(t,t_0,x_0)$   denote the solution of~\eqref{eq:fdyn}
 at time~$t \geq t_0$ for  the initial condition~$x(t_0)=x_0$.
 For the sake of simplicity, we assume from here on
 that~$x(t,t_0,x_0)$ exists and is unique for all~$t \geq t_0\geq 0$
 and all~$x_0 \in \Omega$.

The system~\eqref{eq:fdyn} is said to be 
  \emph{contractive} on~$\Omega$ with respect to  a vector norm $|\cdot| :\R^n \to \R_+$
 if there exists~$\eta>0$
 such that
 \be\label{eq:contdef}
            |x(t,t_0,a)-x(t,t_0,b)|  \leq   e^{ -   (t-t_0) \eta } |a-b|
 \ee
 for all $t\geq t_0\geq   0 $ and all~$a,b \in \Omega$.
This means that  any two trajectories approach 
 one another at an exponential rate~$\eta$. This implies
 in particular that the initial condition is ``quickly forgotten''.

Note that  
contraction can be defined in a more general way, for example   with respect to a time- and space-varying norm~\cite{LOHMILLER1998683}
(see also~\cite{contra_sep}).
We focus here on exponential contraction with respect to a \emph{fixed} vector norm because there exist  easy to check 
sufficient conditions, based on matrix measures, guaranteeing that~\eqref{eq:contdef} holds.  
A vector norm~$|\cdot|:\R^n\to \R_+$ induces a \emph{matrix measure}~$\mu:\R^{n\times n}\to\R$ defined  by
  \begin{align*}
				\mu(A):=\lim_{\varepsilon \downarrow 0} \frac{1}{\varepsilon} (||I+\varepsilon A ||-1) ,    
  \end{align*}
where~$||\cdot||:\R^{n\times n}\to\R_+$ is the matrix norm induced by~$|\cdot|$. For example, for the~$\ell_1$ vector norm, denoted~$|\cdot|_1$,
the induced matrix norm is    the maximum absolute column sum of the matrix,
and the induced matrix measure is
  \begin{align*}
      \mu_1(A)=\max\{c_1(A),\ldots, c_n(A)\},
  \end{align*}
 where
  \begin{align*}
    c_j(A):=A_{jj}+\sum_{ \substack { 1\leq i \leq n\\  i \not = j} } |A_{ij}|  ,
  \end{align*}
  i.e.,   the sum of the entries in column~$j$ of~$A$,
 with non-diagonal
elements   replaced by their absolute values.
Matrix measures satisfy several useful properties (see, e.g.~\cite{vid,Desoer_cont}). 
We list here two properties that will be used later on:
\begin{align*} 
							\mu(A+B)&\leq \mu(A)+\mu(B) ,   &\text{(subadditivity)},  \\
							\mu(cA )&=c\mu(A) \text{ for all }c\geq 0 ,   &\text{(homogeneity)}. 
\end{align*}

If the  Jacobian of~$f$ satisfies 
\be\label{eq:conmm}
\mu(J(t,x))\leq - \eta,\quad \text{for all }x\in\Omega \text{ and all } t\geq  t_0 \geq 0,
\ee
 then~\eqref{eq:contdef} holds (see~\cite{entrain2011} for a self-contained proof).
This is in fact a particular case of using a Lyapunov-Finsler
 function to prove contraction~\cite{contra_sep}.
We will focus on the case where~$\eta>0$, but some of  our results hold when~$\eta\leq 0$ as well.
In this case,~\eqref{eq:contdef} provides a bound on how quickly can trajectories of~\eqref{eq:fdyn}
 separate from one another. 

Often it is useful to work with scaled vector norms (see, e.g.~\cite{sandberg78, Coogan:2016kx}).
 Let~$|\cdot |_*:\R^n\to\R_+$ be some  vector norm, and let~$\mu_*:\R^{n\times n}\to\R$
denote its induced matrix measure. If~$D\in\R^{n\times n}$ is an invertible matrix, and
$|\cdot|_{*,D}  : \R^n \to \R_+$ is the vector norm defined by
$|z|_{*,D}:=|D  z|_* $, then the
induced matrix measure is $
                    \mu_{*,D}(A) = \mu_*(DAD^{-1}).$
For example, the matrix measure induced by the Euclidean norm~$|\cdot|_2$ is
$\mu_2(A)=\max \{\lambda: \lambda \text{ is an eigenvalue of } (A+A')/2\}$, so
\begin{align}\label{eq:scnorm}
\mu_{2,D}(A)&= \mu_{2 }(DAD^{-1})\\
\nonumber            &=\max \{\lambda: \lambda \text{ is an eigenvalue of } (  DAD^{-1}+(DAD^{-1})'   )/2\}.
\end{align}

The next result describes an ISS property of contractive systems with an additive input. 
\begin{Theorem}~\cite{Desoer_cont}
\label{thm:1}
Consider the system
\be\label{eq:withcont}
\dot x(t)=f(t,x(t))+u(t),
\ee
where~$y\to f(t,y)$ is~$C^1$ for all~$t\geq t_0$.
Fix some vector norm~${|\cdot|}:\R^n\to\R_+$ and suppose that~\eqref{eq:conmm}  holds for
 the induced matrix measure~$\mu(\cdot)$. Then the solution of~\eqref{eq:withcont}
with~$x(t_0)=x_0$ satisfies
\begin{align*}
   |x(t,t_0,x_0)|\leq e^{-\eta(t-t_0)}|x_0|+\int_{t_0}^t  e^{-\eta(t-s)} |u(s) | \,ds
\end{align*}
for all~$t\geq t_0$. 
\end{Theorem}

Ref.~\cite{lars16} has applied the ISS property to derive a bound on the error between trajectories of the contractive system~\eqref{eq:fdyn}
and those of a 
 ``simpler'' dynamical system~$\dot y=g(t,y(t))$.
For such a system, pick~$y_0\in\Omega$, and let~$\tau \geq t_0$ be such that the solution~$y(t,t_0,y_0)$ belongs to~$\Omega$
for all~$t\in[t_0,\tau]$. Then
 the difference between the trajectories of the 
two systems
$
d(t):=x(t,t_0,x_0)-y(t,t_0,y_0)
$
satisfies 
 \begin{align}\label{eq:err_bnd}
			|d(t  )| &\leq  e^{-\eta (t-t_0)}|x_0-y_0| +\int_{t_0}^t e^{-\eta (t-s)}  |  f(s,y(s,t_0,y_0))-g(s,y(s,t_0,y_0))   | \diff s 
\end{align}
for all~$t\in [t_0,\tau]$. 
The proof of this result is based on noting that
\begin{align*} 
\dot d(t) 
              &=f(t,x(t))-f(t,y(t))+f(t,y(t))-g(t,y(t)) \nonumber\\
              &=M(t)d+u(t),\nonumber 
\end{align*}
where~$M(t):=\int_0^1 J(t,sx(t)+(1-s)y(t)) \diff s$, and~$u(t):=f(t,y(t))-g(t,y(t))$.
Since~$y(t) \in \Omega$ for all~$t\in[0,\tau]$ 
and~$\Omega$ is convex,~$ sx(t)+(1-s)y(t) \in \Omega$ for all~$t\in[0,\tau]$ and all~$s \in [0,1]$. 
Using~\eqref{eq:conmm} and the 
subadditivity of matrix measures, which, by continuity, extends
to integrals  yields
$\mu(M(t))\leq-\eta$ for all~$t\in[0,\tau]$. Summarizing, $\dot d(t) =M(t)d(t)+u(t)$ is a contractive system with an additive ``disturbance''~$u$ and applying 
   the~ISS property of contractive systems   yields~\eqref{eq:err_bnd}.

Note that the integrand in~\eqref{eq:err_bnd}
  depends on the difference between the vector fields~$f$ and~$g$ 
\emph{evaluated along the trajectory of the~$y$ system}. This is useful, for example, when the trajectory of the~$y$ system 
is    explicitly known.

The applications studied in~\cite{lars16} were   
 contractive systems with time-invariant vector fields approximated by time-invariant~LTI systems. 
Here, we consider a different  case, namely,
 when  the vector field~$f(t,x)$ is time-varying and~$T$-periodic for some~$T>0$, that is,
\[
f(t,z)=f(t+T,z)
\]
 for all~$t\geq t_0$ and all~$z\in\Omega$. 
 It is well-known that in this case every trajectory of~\eqref{eq:fdyn}
 converges to a unique periodic solution~$\gamma(t)$ of~\eqref{eq:fdyn} with period~$T$
(see~\cite{ entrain2011} for a self-contained proof). This entrainment
property is very important in applications (see, e.g.~\cite{RFM_entrain,entrain2011}). 
However, 
the proof of entrainment is based on implicit arguments and provides no information on the properties of
the period trajectory (except for its period). 
Our goal here  is to develop a suitable bound   
for the difference between  $\gamma(t)$ and the periodic solution~$\kappa(t)$
 of some simpler 
  approximating~$y$ system, and to suggest suitable approximating systems. We also
	show that these explicit bounds can be used to derive new theoretical results on the response of contractive systems
	to a sinusoidal input. 
	The next three  sections present our main results.
\section{Bounds on the difference between two periodic trajectories} \label{sec:bounds} 
In this section, we consider the $T$-periodic orbit
of a ~$T$-periodic contractive system.
Theorem~\ref{thm:periodbound} below is our main result in this section,
and provides a bound on the distance of this periodic orbit to a $T$-periodic orbit of some approximating system.

\begin{Theorem}  \label{thm:periodbound}
Consider the system
\begin{align}  \label{eq:14}
  \dot{x}=f(t,x)
\end{align}
whose trajectories evolve  
  on a  compact and convex state-space~$\Omega\subseteq \mathbf{\R}^n$. Suppose that~$f(t,x)$ is $T$-periodic and that~$f(t,x)$ and~$J(t,x)$ are continuous in $(t,x)$. Let $|\cdot|$ be some vector norm on $\mathbb{R}^n$ and $\mu(\cdot)$ its induced matrix measure, and suppose that~$\mu(J(t,x))\leq -\eta<0$   for all $t\geq 0$ and all~$x\in\Omega$. Let $\gamma(t)$ be the unique periodic trajectory of \eqref{eq:14} with period $T$. Consider another time-varying system
\begin{align} \label{eq:13}
\dot{y}=g(t,y)
\end{align}
and suppose that $g(t,y)$ is also $T$-periodic and that~$\kappa(t)$ is a $T$-periodic trajectory of \eqref{eq:13} 
with~$\kappa(t)\in \Omega$ for all~$t\in[0,T]$. Define~$c:\R_+ \to \R_+$ by
\begin{align}
  \label{eq:21}
  c(\alpha):=\int_{0}^\alpha e^{-\eta(\alpha-s)}|f(s,\kappa(s))-g(s,\kappa(s))| \diff s.
\end{align}
Then  the difference between the two periodic trajectories satisfies
\begin{align} \label{eq:16}
|\gamma(\tau)-\kappa(\tau)|\leq \frac{ e^{-\eta \tau}   }{1-e^{-\eta T}} c(T)  +c(\tau)
\end{align}
  for all $\tau\in[0,T]$.
\end{Theorem}

Note that the bound here depends on the difference between the vector fields~$f$ and~$g$ evaluated
 along the periodic trajectory~$\kappa(s)$
of the ``simpler'' $y$ system. This is useful for example when the~$y$ system is an asymptotically stable 
LTI system with a sinusoidal forcing term, as then~$\kappa(t)$ is known explicitly. 

\begin{proof}
Define the $T$-periodic function $  u(t):= f(t,\kappa(t))-g(t,\kappa(t))$. 
For any~$t\geq 0$, 
 Theorem \ref{thm:1} gives
  \begin{align} \label{eq:17}
    |\gamma(t)-\kappa(t)|\leq e^{-\eta t}|\gamma(0)-\kappa(0)| +\int_{0}^t e^{-\eta (t-s)}  | u(s)   | \diff s.
  \end{align}
 There exist~$\tau\in[0,T)$ and a non-negative integer~$k $ such that~$t=kT+\tau$.
  Observe that~$ \gamma(t)-\kappa(t) =     \gamma(\tau)-\kappa(\tau) $.
Write the integral on the right-hand side of~\eqref{eq:17} as 
\begin{align*} 
 \int_{0}^{t} e^{-\eta (t-s)}  |u(s) | \diff s &=  e^{-\eta \tau}\int_{0}^{kT+\tau} e^{-\eta (kT-s)}  |  u(s)  | \diff s \\
\nonumber&=e^{-\eta \tau}\sum_{i=0}^{k-1}\int_{iT}^{(i+1)T}e^{-\eta (kT-s)}  |  u(s)  | \diff s+\int_0^\tau e^{-\eta(\tau-p)}|u(p)|\diff p .
\end{align*}
Moreover,
\begin{align*}
  \sum_{i=0}^{k-1}\int_{iT}^{(i+1)T}e^{-\eta (kT-s)}  |  u(s)  | \diff s & =c(T) \frac{1-e^{-\eta k }}{1-e^{-\eta T} }
\end{align*}
so that
\begin{align*}  
|\gamma(\tau)-\kappa(\tau)|&  = |\gamma(t)-\kappa(t)|\leq  e^{-\eta (kT+\tau)}|\gamma(0)-\kappa(0)|+   e^{-\eta \tau} c (T)   \frac{1-e^{-k\eta T}}{1-e^{-\eta T} }   +c(\tau)  
\end{align*}
for all $\tau\in[0,T)$. Taking $k\to \infty$ completes  the proof.
\end{proof}

Note that the bound is actually based on taking the time~$t\to \infty$. This is possible because we are considering the
difference between two   periodic trajectories.

The next example is important, as it
shows that in general the bound~\eqref{eq:16}  cannot be improved.
\begin{Example}
Consider the scalar system
\be\label{eq:scals}
\dot x=f(t,x):=-x+ 1+\sin(2\pi t/T),
\ee
 with~$T>0$. Note that~$\Omega:=[0,2]$ is an invariant set of this dynamics, and 
  that~$f$ is~$T$-periodic. The Jacobian of~$f$ is~$J(x)=-1$, so for any vector norm the induced matrix measure satisfies~$\mu(J(x))=-1$. 
	For any initial condition  the solution of~\eqref{eq:scals} converges to the $T$-periodic trajectory:
	\be\label{eq:yhngg}
	\gamma(t)=1+\frac{T^2\sin(2\pi t/T)-2\pi T \cos(2\pi t/T)}{4\pi^2+T^2}.
		\ee
	
	Consider the approximating system~$\dot y=-y$, which is (vacuously)  $T$-periodic, and admits
	the~$T$-periodic solution~$\kappa (t)\equiv 0$, that belongs to~$\Omega$ for all~$t$. 
In this case, \eqref{eq:21} yields 
\begin{align*}
  c(\alpha)&=\int_{0}^\alpha e^{- (\alpha-s)}|1+\sin(2\pi s/T) | \diff s\\
\nonumber	&=	 1-e^{-\alpha}+\frac{2\pi T e^{-\alpha}-2\pi T \cos(2\pi \alpha/T)+T^2\sin(2\pi\alpha/T)}{4\pi^2+T^2}.
\end{align*}
Thus for large values of~$\alpha$,
\begin{align}\label{eq:clatggv}
  c(\alpha)&\approx  
	 \frac{ -2\pi T \cos(2\pi \alpha/T)+T^2\sin(2\pi\alpha/T)}{4\pi^2+T^2} +1 .
\end{align}
Now consider the case where~$T\to\infty$ and~$\tau=T-\varepsilon$, with~$\varepsilon>0$ and very small.
Then~\eqref{eq:yhngg} implies that
 the term on the left-hand side of~\eqref{eq:16} is
 \[
	|\gamma(\tau) | \approx 1+\sin(2\pi (T-\varepsilon)/ T),   
 \]
whereas~\eqref{eq:clatggv} implies that 
 the term on the right-hand side of~\eqref{eq:16} is
	\begin{align*}
	\frac{ e^{-\eta \tau}   }{1-e^{-\eta T}} c(T)  +c(\tau)&\approx c(\tau)\\
   &\approx  
	   1+\sin(2\pi (T-\varepsilon)/T) .
\end{align*}
Thus, this example shows that in general the bound~\eqref{eq:16}  cannot be improved.
\qedwhite
\end{Example}

We now derive a simpler (and less tight) bound. 
By the definition of~$c(\cdot)$,
\begin{align*}
   c(\alpha) \leq    \frac{1-e^{-\eta \alpha} }{\eta}  \max_{t\in[0,\alpha] } |f(t,\kappa(t))-g(t,\kappa(t))|.
\end{align*}
for all~$\alpha\geq 0$, and combining this 
with~\eqref{eq:16} yields  the following result.
\begin{Corollary}
\label{cor:maincor}
Under the hypotheses of Theorem \ref{thm:periodbound},
\begin{align*}  
  |\gamma(\tau)-\kappa(\tau)|\leq \frac{ 1}{\eta}
		\max_{t\in[0,T] } |f(t,\kappa(t))-g(t,\kappa(t))|  
\end{align*}
for all~$\tau\geq 0$.
\end{Corollary}

This bound is useful in cases where one can  establish a bound on  the 
difference between the vector fields~$f$ and~$g$
 along the periodic trajectory~$\kappa$ of the approximating system.
Note that the bound here demonstrates a clear tradeoff: if~$g$ is ``close'' to~$f$ then the error~$f-g$ will
be small, yet~$\kappa$ may be an unknown 
 complicated trajectory (as we assume that~$f$ is a nonlinear vector filed). 
On the other hand, if~$g$ is relatively simple (e.g., the vector field of an~LTI system) then 
$\kappa$ may be known explicitly yet that difference~$|f-g|$ may be large.

To summarize,  Theorem \ref{thm:periodbound} and Corollary \ref{cor:maincor} provide a bound on the distance of the unique~$T$-periodic trajectory of a contractive  system 
 and some $T$-periodic trajectory of an  approximating system. The next step is 
to determine  a suitable approximating system. 
We propose two natural approximating systems for the
case where  the periodic vector field arises via a periodic forcing function.
 The first approximating system considers the time-averaged periodic forcing function to arrive at an autonomous dynamical system with 
a unique equilibrium. The second approximating system results from a linearization of the dynamics, keeping the periodic excitation as is.

\section{Approximating Systems} \label{sec:approx}
From hereon, we consider a special case of the contractive system~\eqref{eq:14} with the form
\begin{align*}
  \dot{x}(t)= f(t,x(t))=F(x(t),u(t))
\end{align*}
where $u(t)$ is a given $m$-dimensional, $T$-periodic excitation.

\subsection{Averaging the input} 
Our first result is based  
on using a ``simpler''~$y$ system derived by averaging   the  excitation~$u$ over a period. 
The excitation in the~$y$ system
 is thus constant.
We assume that the $y$ system admits an equilibrium point~$e\in\Omega$, and 
apply Theorem~\ref{thm:periodbound} to derive a 
  bound on the distance between the periodic trajectory~$\gamma(t)$
	of the original~$x$ system and the point~$e$.

\begin{Theorem}\label{thm:averexci}
Consider the system
\be\label{eq:capitalfdynu}
\dot x=F(x,u),
\ee
where~$u$ is  an~$m$-dimensional 
 periodic excitation  with period~$T\geq 0$. 
Suppose that the trajectories of~\eqref{eq:capitalfdynu} evolve on a compact and convex state space~$\Omega\subset\R^n$. 
 Assume that for some vector norm~$|\cdot|:\R^n\to\R_+$ and  induced
 matrix measure~$\mu:\R^{n\times n }\to\R$,
 \begin{align*}
				\mu\left(\frac{\partial F}{\partial x} (x,u(t)) \right)\leq -\eta<0   
 \end{align*}
for all~$t\geq 0$ and all~$x \in \Omega$. %
Let $\gamma(t)$ be the unique, attracting, $T$-periodic orbit of~\eqref{eq:capitalfdynu} in~$\Omega$. Then, for any $z\in \Omega$,
 \begin{equation}
\label{eq:err_bnd_per}
			 |x(t,0,z)-z | \leq   
			\int_{ 0}^t e^{-\eta (t-s)}  |  F(z,u(s))   | \diff s  
 \end{equation}
for all~$t\geq 0$. In particular,  for all~$t\geq 0$,
\begin{align*}
			 |x(t,0,z)-z | \leq    (1-e^{-\eta t})c/{\eta},  
\end{align*}
 where~$c:=\max_{t\in[0,T]}|  F(z,u(t))   | $.
Moreover, for all $\tau\in[0,T]$, 
\begin{align}
\label{eq:20}|\gamma(\tau)-z|&\leq \frac{e^{-\eta \tau}}{1-e^{-\eta T}}\int_0^T e^{-\eta(T-s)}|F(z,u(s))|\diff s+\int_0^\tau e^{-\eta(\tau-s)}|F(z,u(s))|\diff s \\
\label{eq:20-2}   &\leq c/\eta.
\end{align}

\end{Theorem}

\begin{proof}
Define the approximating system $\dot{y}=G(y)\equiv 0$ so that any $z\in\Omega$ is an equilibrium. Then~\eqref{eq:err_bnd_per}    follows  from 
 the bound \eqref{eq:err_bnd}, and 
the bounds~\eqref{eq:20} and \eqref{eq:20-2}   follow from Theorem~\ref{thm:periodbound} and Corollary~\ref{cor:maincor}.
\end{proof}

The following simple example demonstrates a case where
the bounds in Theorem~\ref{thm:averexci} are tight.
\begin{Example}
  Consider the scalar system~$\dot{x}=F(x,u):=-ax+b$ with~$a>0$. Then~$\gamma(t)\equiv b/a=:e$ is a periodic trajectory. 
	This system is contracting with rate~$\eta=a$. The bound~\eqref{eq:20-2} gives $|e-z|\leq |-az+b|/a=|e-z|$ for any $z\in\mathbb{R}$ so that this bound is tight.\qedwhite
\end{Example}

The next two examples  demonstrate that a natural choice for $z$ in Theorem~\ref{thm:averexci} is the equilibrium point induced by the average of the periodic excitation.
\begin{Example}\label{exa:linear}
Our focus here is on nonlinear dynamical systems, but it is still useful  to begin by considering  the  
  linear   system
\be\label{eq:linab}
			\dot x = A x+ B u,
\ee
where~$A\in\R^{n\times n}$ is Hurwitz, $B\in\R^{n\times m}$, and~$u$ is an $m$-dimensional
 $T$-periodic control.
It is well-known that such a system is contractive~\cite{sontag_cotraction_tutorial,strom_log_norms}. For the sake of completeness we repeat the argument here. 
We use the notation~$Q>0$ to  denote that a matrix~$Q$ is symmetric and positive-definite. Since~$A$ is Hurwitz, there exist~$\eta >0$ and~$Q>0$ such that
\be\label{eq:paq}
					QA+A'Q\leq -2\eta Q.
\ee
Let~$P>0$ be a  matrix such that~$P^2=Q$. Then multiplying~\eqref{eq:paq} by~$P^{-1}$ on the left and on the right yields
\be\label{eq:pminleft}
		PAP^{-1}+P^{-1}A' P \leq -2\eta I.
\ee
This means that the Jacobian~$A$ of~\eqref{eq:linab} satisfies~$\mu_{2,P}(A)\leq -\eta$,
where~$\mu_{2,P}$ is the matrix measure induced by the scaled Euclidean norm~$|z|_{2,P}:=|Pz|_2$ (see~\eqref{eq:scnorm}).
Thus,~\eqref{eq:linab} is contractive with respect to this scaled norm with contraction rate~$ \eta$, 
and every solution of~\eqref{eq:linab} converges to the unique  $T$-periodic solution~$\gamma(t)$ of~\eqref{eq:linab}.
Let~$\bar u:=\frac{1}{T}\int_0^T u(s)ds$ and choose $z=A^{-1}B\bar{u}=:e$, the equilibrium of the time-invariant system with input equal to $\bar{u}$.
To apply the bound~\eqref{eq:20-2}, note that 
\begin{align*}
				F(e,u(s))=Ae+B u(s)=B(u(s)-\bar u) .  
\end{align*}
Thus,~$|  F(e,u(s))   |_{2,P}= \left(  (u(s)-\bar u) ' B'P'PB (u(s)-\bar u)  \right)^{1/2} $,
and  the  bound~\eqref{eq:20-2} yields
\begin{align} \label{eq:blins}
  |\gamma(\tau)-e|_{2,P}
   &\leq \frac{1}{\eta}  \max_{t\in[0,T]}    \left(  (u(t)-\bar u) ' B'P'PB (u(t)-\bar u)  \right)^{1/2}
\end{align}
for all $\tau\in[0,T]$.

Of course, for linear systems   the periodic solution corresponding to sinusoidal excitations 
is known explicitly    in terms of the system's
 frequency response. Nevertheless,~\eqref{eq:blins}
seems to be new and provides considerable intuition: the bound on the distance between~$\gamma(t)$ and~$e$ 
decreases when: the contraction rate~$\eta$ increases;  
 the input channel~$B$ becomes ``more orthogonal'' to the matrix~$P$
in~\eqref{eq:pminleft}; or~$\max_{t\in[0,T]} |u(t)-\bar u|$  decreases, that is, the periodic excitation
  becomes more similar to its mean.
\qedwhite
\end{Example} 

 The next example demonstrates an application of
Theorem~\ref{thm:averexci}   for a nonlinear contractive
system.

\begin{Example}\label{exa:rfm}
The ribosome flow model~(RFM)~\cite{reuveni}
is a nonlinear compartmental model describing the unidirectional flow  of   particles 
along a 1D chain  of~$n$ sites using~$n$ non-linear first-order differential equations:
\begin{align}\label{eq:rfmeqns}
                    \dot{x}_1&=\lambda_0 (1-x_1) -\lambda_1 x_1(1-x_2), \nonumber \\
                    \dot{x}_2&=\lambda_{1} x_{1} (1-x_{2}) -\lambda_{2} x_{2} (1-x_3) , \nonumber \\
                    \dot{x}_3&=\lambda_{2} x_{ 2} (1-x_{3}) -\lambda_{3} x_{3} (1-x_4) , \nonumber \\
                             &\vdots \nonumber \\
                    \dot{x}_{n-1}&=\lambda_{n-2} x_{n-2} (1-x_{n-1}) -\lambda_{n-1} x_{n-1} (1-x_n), \nonumber \\
                    \dot{x}_n&=\lambda_{n-1}x_{n-1} (1-x_n) -\lambda_n x_n.
\end{align}
Here~$x_i(t) \in[0,1]$ represents the level of occupancy of site~$i$ at time~$t$, normalized such that~$x_i(t)=1$ [$x_i(t)=0$]
means that site~$i$ is completely full [empty]. 
The state-space is thus~$[0,1]^n$, and this  is an invariant set of~\eqref{eq:rfmeqns} (see~\cite{RFM_entrain}).
The transition
rate~$\lambda_i>0$ controls the flow from site~$i$ to site~$i+1$, with~$\lambda_0$ [$\lambda_n$]
called the initiation [exit] rate. To understand these equations, note that they may be written
as~$\dot{x}_i=g_{i-1}(x)-g_i(x)$, 
where~$g_k(x)$ is the flow from site~$k$ to site~$k+1$ at time~$t$.
This flow increases with~$x_k$ and decreases with~$x_{k+1}$.
In other words, the  flow satisfies a ``soft'' excluded volume 
principle: as site~$k+1$ becomes fuller the flow from site~$k$ to site~$k+1$ decreases. This  models  the fact that the particles have volume and thus cannot overtake one another. The rate at which particles
 leave the chain, that is,
$R(t):=\lambda_n x_n$ is called the  \emph{production rate}.

The RFM with~$n>2$ is actually not contractive in the sense defined above on~$[0,1]^n$, as there exists~$p\in [0,1]^n$
such that~$J(p)$ is singular, 
but it is ``weakly'' contractive in a well-defined sense; see~\cite{cast_book,3gen_cont_automatica}.

Recently, the RFM has been used to model and analyze the flow of ribosomes (the  particles) along groups of codons 
(the sites) along  the mRNA molecule during translation
(see, e.g.~\cite{zarai_infi,RFM_feedback,RFM_stability,RFM_entrain,RFMR,RFM_model_compete_J,RFM_sense,rfm_max_density,rfm_opt_down}). 
In this case, every ribosome that leaves the chain releases the produced protein, so~$R(t)$ is the protein production rate at time~$t$. 
The values of the transition rates depend on various biophysical properties, e.g. the abundance of tRNA  molecules 
 that carry the corresponding amino-acids.

Consider the RFM with~$n=2$ and a \emph{time-varying}
 initiation rate~$u_0(t)$, that is,
\begin{align}\label{eq:rfm2}
					\dot x_1&=  (1-x_1) u_0-\lambda_1 x_1 (1-x_2)\nonumber,\\
					\dot x_2&= \lambda_1 x_1 (1-x_2)-\lambda_2 x_2,
\end{align}
where~$\lambda_1,\lambda_2$ are positive constants. Suppose that~$u_0(t)=\lambda_0+\sin(2\pi t/T)$, 
with~$\lambda_0>1$, $T>0$, i.e. the initiation rate is 
  a strictly positive   periodic function with (minimal) period~$T$. The state space here is~$\Omega:=[0,1]^2$.  
The Jacobian of~\eqref{eq:rfm2} is 
\begin{align*}
			J(t,x)=\begin{bmatrix} -u_0(t)-\lambda_1(1-x_2) &\lambda_1 x_1 \\
			\lambda_1(1-x_2)& -\lambda_1 x_1-\lambda_2 
			\end{bmatrix}.  
\end{align*}

 The off-diagonal terms are non-negative for any~$x\in [0,1]^2$, 
so~$\mu_1(J(t,x))= \max\{  -u_0 (t),-\lambda_2 \}$ for all~$t\geq 0$ and all~$x\in[0,1]^2$. 
Thus, the system  is contractive with respect to the~$\ell_1$ norm with contraction rate~$\eta:=\min\{ \lambda_0-1 ,\lambda_2 \}>0$. This means that it admits a unique periodic solution~$\gamma \in [0,1]^2$,
with period~$T$,  and that every solution converges to~$\gamma$. Entrainment in mRNA translation is important as biological organisms are often exposed to
  periodic excitations, for example the periodic cell-cycle division process.  
	Proper biological functioning  requires entrainment to such excitations~\cite{RFM_entrain}.

Let $\bar u_0=\frac{1}{T} \int_0^T u_0(s) \diff s =\lambda_0$ and consider the system
\begin{align}\label{eq:rfm3}
					\dot y_1&=\lambda_0   (1-y_1)-\lambda_1 y_1 (1-y_2)\nonumber,\\
					\dot y_2&= \lambda_1 y_1 (1-y_2)-\lambda_2 y_2.
\end{align}
This system  admits an equilibrium point
\be\label{eq:e2exp}
			e=\begin{bmatrix} \frac{ \lambda_0\lambda_1-  \lambda_0 \lambda_2 - \lambda_1 \lambda_2 +\sqrt{d}}{2
			 \lambda_0 \lambda_1}  &
				\frac{  \lambda_0\lambda_1 + \lambda_0 \lambda_2 +\lambda_1 \lambda_2-\sqrt{d} }{2\lambda_1 \lambda_2} \end{bmatrix}' \in (0,1)^2,
\ee
where~$d:=    4  \lambda_0^2 \lambda_1 \lambda_2 +( \lambda_0  \lambda_1-\lambda_0
 \lambda_2-\lambda_1 \lambda_2   )^2     $. 

Here,
\begin{align*}
F(e,u(s))=\begin{bmatrix} (\lambda_0+\sin(2\pi s/T)) (1-e_1 )-\lambda_1 e_1 (1-e_2 ) \\
					 \lambda_1 e_1 (1-e_2)-\lambda_2 e_2  \end{bmatrix},  
\end{align*}
and since~$e$ is an equilibrium point of~\eqref{eq:rfm3},
 $F(e,u(s))=\begin{bmatrix}   (1-e_1 )\sin(2\pi s/T) & 0  \end{bmatrix}'$. 
Thus, the bound~\eqref{eq:err_bnd_per} yields
\begin{align}\label{eq:bdee}
			 |x(t,0,e)-e | _1 &\leq   
		  (1-e_1) \int_{ 0}^t e^{-\eta (t-s)}  |   \sin(2\pi s/T)   |  \diff s.
\end{align}
  Likewise, \eqref{eq:20} implies
\begin{align}\label{eq:30}
   |\gamma(\tau)-e|_1&\leq (1-e_1) \frac{e^{-\eta \tau}}{1-e^{-\eta T}}\int_0^T e^{-\eta(T-s)}| \sin(2\pi s/T)  |\diff s+(1-e_1) \int_0^\tau e^{-\eta(\tau-s)}| \sin(2\pi s/T)  |\diff s
\end{align}
for all $\tau$.
  Furthermore, 
\begin{align*}
  |F(e,u(t))|_1=(1-e_1) |   \sin(2\pi t/T)   | \leq 1-e_1
\end{align*}
so  \eqref{eq:20} implies the simpler yet
more conservative bound 
\be\label{eq:symcb}
|\gamma(\tau)-e|_1\leq (1-e_1)/\eta
\ee
  for all $\tau$. 

Note that all the bounds above can be computed \emph{explicitly}. For example, a tedious yet straightforward calculation of the 
integrals in~\eqref{eq:30} yields 
\begin{align*}
  |\gamma(\tau)-e|_1&\leq    \frac{ 2\pi T(1-e_1)  \coth(\eta T/4) } { e^{\eta \tau} (4\pi^2+\eta^2 T^2) } 
	+\frac{T(1-e_1)}{4\pi^2+\eta^2 T^2} r(\tau),
\end{align*}
where
\begin{align*}
\nonumber	&r(\tau):=\begin{cases}
  2\pi e^{-\eta \tau}+\eta T\sin(2\pi\tau/T)-2\pi\cos(2\pi\tau/T)&\text{if }0\leq \tau<T/2,\\
     2\pi e^{-\eta \tau } ( 1+2 e^{\eta T/2} )-\eta T \sin (2\pi \tau/T)+2\pi \cos( 2 \pi \tau/T )&\text{if }T/2 \leq \tau <  T.
 \end{cases}
\end{align*}
Summarizing, in this example, we have an analytical expression both for~$e$ and for the error bounds. Taken together, this
provides considerable explicit information on the periodic trajectory~$\gamma$.

For the case~$\lambda_0=4$, $\lambda_1=1/2$, $\lambda_2=4$, and~$T=2$,
 Fig.~\ref{fig:rfm_cyc} 
depicts the  periodic trajectory of~\eqref{eq:rfm2} and the equilibrium point~$e$ of~\eqref{eq:rfm3}, and
 Fig.~\ref{fig:err_p_bnd}  depicts the error 
			$ |x(t,0,e)-e | _1$ and the bound~\eqref{eq:bdee}. 
			In this case, \eqref{eq:e2exp} yields~$e\approx\begin{bmatrix} 0.8990&
    0.1010 \end{bmatrix}'$ (all numerical values in this note are to four digit accuracy). Fig.~\ref{fig:err_period} illustrates the other
		bounds on the periodic trajectory.
		It may be seen   that these bounds indeed provide
		a reasonable approximation for the~$\ell_1$ distance between
		the unknown periodic trajectory and the point~$e$.  
\qedwhite
\end{Example}

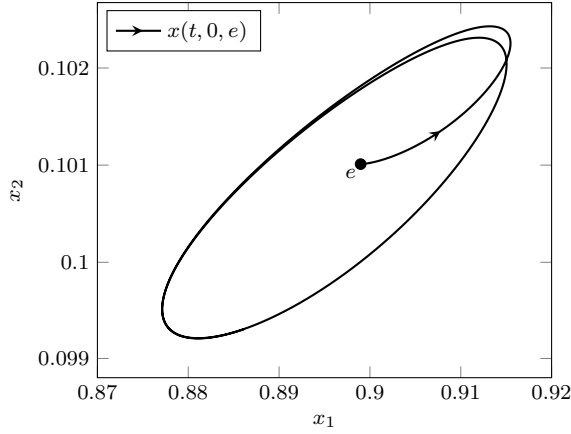
\begin{figure}[t]
{\footnotesize
  \begin{center}
\pgfplotsset{scaled y ticks=false}
\pgfplotsset{compat=1.12}
\begin{tikzpicture}
    \begin{axis}[width=3in,
   xmin=0.87, xmax=0.92, ymax=.1027,ymin=.0988,
    xlabel={$x_1$},
    ylabel={$x_2$},
    ytick={.0990,.100,...,.103},
scaled ticks = false, 
x tick label style={
        /pgf/number format/.cd,
        precision=2,
    },
y tick label style={
        /pgf/number format/.cd,
        precision=4,
    },
    legend style={at={(0.02,.98)},anchor=north west},
    ]
    \addplot[black, line width=.8pt, forget plot] file{Figures/ex1_traj.txt} [arrow inside={}{0.05}];
    \addplot[black, mark=*, forget plot] coordinates {(0.8989794856,	0.1010205144)};
    \node[label={[label distance=-6pt]south west:$e$}] at (0.8989794856,	0.1010205144) {};
\addlegendimage{arrow inside={}{0.6}, line width=.8pt};
\addlegendentry{$x(t,0,e)$};
\end{axis}
\end{tikzpicture}
\end{center}
}
\caption{ RFM in Example~\ref{exa:rfm}:
 Averaging the periodic excitation leads to an autonomous approximating system with a unique equilibrium~$e$.  For the periodic excitation the trajectory $x(t,0,e)$ converges to the unique periodic solution~$\gamma(t)$  of the~RFM,
 and Theorem~\ref{thm:1} provides a bound on the distance between~$x(t,0,e)$ and~$e$ for all~$t\geq 0$.}
\label{fig:rfm_cyc} 
\end{figure}
\begin{figure}[t]
{\footnotesize
  \begin{center}
\pgfplotsset{scaled y ticks=false}
\pgfplotsset{compat=1.12}
\begin{tikzpicture}
    \begin{axis}[width=3in, height=2in,
   xmin=0, xmax=6, ymax=.03,ymin=0,
    xlabel={$t$},
    legend style={at={(1.02,.98)},anchor=north west},
y tick label style={
        /pgf/number format/.cd,
        precision=4,
    },
    ]
    \addplot[black, line width=.8pt,] file{Figures/ex1_exact_diff.txt};
    \addplot[black, dashed, line width=.8pt,] file{Figures/ex1_bound.txt};
\end{axis}
\end{tikzpicture}
\end{center}
}
\caption{ The error~$|x(t,0,e)-e|_1$ (solid line) and the bound provided by~\eqref{eq:bdee} (dashed line) as a function of time for Example~\ref{exa:rfm}.} \label{fig:err_p_bnd}
\end{figure}
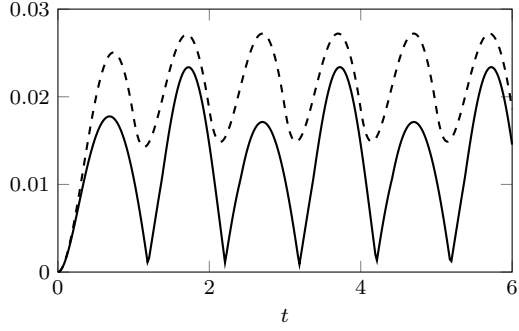

\begin{figure}[t]
{\footnotesize
  \begin{center}
\pgfplotsset{scaled y ticks=false}
\pgfplotsset{compat=1.12}
\begin{tikzpicture}
    \begin{axis}[width=3in, height=2in,
   xmin=0, xmax=2, ymax=.04,ymin=0,
    xlabel={$t$},
    xtick={0,1,2},
    xticklabels={$0$,$\frac{1}{2}T$, $T$},
    legend style={at={(1.02,.98)},anchor=north west},
y tick label style={
        /pgf/number format/.cd,
        precision=4,
    }
    ]
    \addplot[black, line width=.8pt] file{Figures/ex1_periodic_traj.txt};
    \addplot[black, dashed, line width=.8pt] file{Figures/ex1_periodic_bound.txt}; 
    \addplot[black, dotted, line width=.8pt] coordinates{(0,0.033673504811215) (2,0.033673504811215)};
\end{axis}
\end{tikzpicture}
\end{center}
}
\caption{
RFM in Example~\ref{exa:rfm}. The error~$|\gamma(t)-e|_1$ (solid line) and the bounds in 
Theorem~\ref{thm:averexci}:   the bound~\eqref{eq:30}  (dashed line) and the bound~\eqref{eq:symcb}
(dotted line). These bounds can be obtained analytically for this example.}
\label{fig:err_period}
\end{figure}
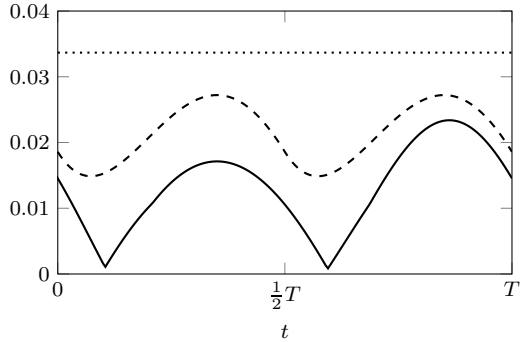

Thm.~\ref{thm:averexci} is based on 
  averaging   the  excitation over a period, thus obtaining a constant input. 
Such an approximation is not always suitable. 
For example, when~$u(t)=\sin(2\pi t/T)$ then~$\bar u:=\frac{1}{T}\int_0^T u(t) dt=0$
for all~$T$.
					This may obscure the effect of the \emph{frequency} of the excitation in the derived bounds.
	The approach in the next subsection 
	  tries to overcome this using a
	different approximating  system, namely,  an LTI system  that is
		excited by the original periodic
	input. 

	\subsection{An LTI approximation}
\begin{Theorem}\label{thm:linearized}
Consider the system
\be\label{eq:fdynu}
\dot x=F(x,u),
\ee
where~$u$ is  an~$m$-dimensional 
 periodic excitation  with period~$T >  0$. 
Suppose that the trajectories of~\eqref{eq:fdynu} evolve on a compact and convex state space~$\Omega\subset\R^n$. 
 Assume that for some
 vector norm~$|\cdot|:\R^n\to\R_+$ and  the  induced
 matrix measure~$\mu:\R^{n\times n }\to\R$,
 \begin{align*}
			\mu\left (	\frac{\partial F}{\partial x} (x,u(t)) \right  )\leq -\eta<0   
 \end{align*}
for all~$t\geq 0$, all~$x \in \Omega$. %
Let $\gamma(t)$ be the unique, attracting, $T$-periodic orbit of~\eqref{eq:fdynu} in~$\Omega$.

Suppose also that for the unforced dynamics,
i.e.~$\dot x=F(x,0)$, there exists a locally stable equilibrium
 point~$e \in \Omega$, and
without loss of generality, that~$e=0$. 
Let
  \begin{align*}
  A:=\frac{\partial F}{\partial x} (0,0),\qquad B:=\frac{\partial F}{\partial u}(0,0),
\end{align*}
and consider the~LTI  approximating system 
  \begin{align} \label{eq:ltigu}
				\dot y  =Ay+B u:=G(y,u ).
  \end{align}
  Pick~$x_0, y_0 \in \Omega$ and let~$\tau \geq 0$ be such that~$y(t)\in\Omega$ for all~$t\in[0,\tau]$ where $y(t)$ is the solution of \eqref{eq:ltigu} with $y(0)=y_0$. 
 Then  
 \begin{align*}
&|x(t)-y(t)|\leq e^{-\eta t}|x_0-y_0|+\int_{ 0}^t e^{-\eta (t-s)}  |  F(y(s),u(s))-G(y(s),u(s))   | \diff s  
 \end{align*}
for all~$t \in [0,\tau]$.  
Moreover, let $\kappa(t)$ be the unique $T$-periodic trajectory of \eqref{eq:ltigu} and assume that~$\kappa(t)\in\Omega$ for all~$t$. Then, for all $\tau\in[0,T]$, 
\begin{align}
\nonumber |\gamma(\tau)-\kappa(\tau)|&\leq \frac{e^{-\eta \tau}  }{1-e^{-\eta T}}\int_{0}^T e^{-\eta(T -s)} |F(\kappa(s),u(s))-G(\kappa(s),u(s))| \diff s \\
  \label{eq:26}                            &\quad +\int_{0}^\tau e^{-\eta(\tau-s)} |F(\kappa(s),u(s))-G(\kappa(s),u(s))| \diff s \\
  &\leq \frac{1}{\eta} \max_{t\in[0,T]}|F(\kappa(t),u(t))-G(\kappa(t),u(t))|   .\nonumber
\end{align}
\end{Theorem}

We emphasize again that the advantage of the bounds  here
 is that the integrand depends on the difference between the vector fields~$F$ and $G$
evaluated along the solution~$\kappa$  of the~LTI system \eqref{eq:ltigu}. 
Note that our assumptions imply that~$A$ is Hurwitz and thus, for any initial condition, $y(t)$ converges to the periodic trajectory~$\kappa(t)$.
In some cases,  this solution  can be written explicitly,
 and the integral can be computed explicitly.  For example, if~$u(t)$ is a complex exponential,
 then~$\kappa(t)$ is also a complex exponential and can be 
easily computed using a Fourier transform. Then
 a bound on~$| F(\kappa (t),u(t))-G(\kappa (t),u(t))   |$,   $t\in[0,T]$, may be straightforward to establish.
This leads to the following corollary of Theorem~\ref{thm:linearized}. For the sake of simplicity, we state this for the case of a scalar control.
\begin{Corollary}\label{coro:ftrans}
Consider the system~\eqref{eq:fdynu} 
where~$u(t)=\sum_{i=1}^p a_i \cos(\omega_i t)$ with~$a_i \in \R$ and every~$\omega_i$ 
has the form~$\omega_i=2\pi k_i/T$, with~$k_i$ a non-negative integer.
 Suppose that the trajectories of~\eqref{eq:fdynu} evolve on a compact and convex state space~$\Omega\subset\R^n$. 
 Assume that for some
 vector norm~$|\cdot|:\R^n\to\R_+$ the  induced
 matrix measure~$\mu:\R^{n\times n }\to\R$
satisfies 
\begin{align*}
				\frac{\partial F}{\partial x} \left(x, \sum_{i=1}^p a_i\cos(\omega_i t)\right ) \leq -\eta<0  
\end{align*}
for all~$t\geq 0$ and all~$x \in \Omega$. 
Let $\gamma(t)$ be the unique, attracting, $T$-periodic orbit of~\eqref{eq:fdynu} in~$\Omega$.

Suppose also that   the unforced dynamics,
i.e.~$\dot x=F(x,0)$  admits  a locally stable equilibrium
 point~$e \in \Omega$, and
without loss of generality, that~$e=0$. 
Let
  \begin{align*}
  A:=\frac{\partial F}{\partial x} (0,0),\qquad B:=\frac{\partial F}{\partial u}(0,0),
\end{align*}
and consider the approximating system 
  \begin{align} \label{eq:11}
				\dot y  =Ay+B u:=G(y,u ).
  \end{align} 
	Let~$\hat g(s):=(sI-A)^{-1}b$ and 
let~$\kappa(t)$ be the unique $T$-periodic trajectory of \eqref{eq:11}, that is,
\begin{align*}
  \kappa_r(t)=\sum_{i=1}^p a_i |\hat g_r(j\omega_i)|   \cos(\omega_i t+\angle \hat g_r(j \omega_i)),\quad r=1,\dots, n,
\end{align*}
and assume that~$\kappa(t)\in\Omega$ for all~$t\in [0,T]$. Then for all $\tau\in[0,T]$, 
\begin{align}\label{eq:nbh}
  |\gamma(\tau)-\kappa(\tau)|&
  \leq \frac{1}{\eta} \max_{t\in[0,T]}\left |H \left (\kappa(t),\sum_{i=1}^p a_i\cos(\omega_i t) \right )\right | 
\end{align}
where~$H(z,v):=F(z,v)-G(z,v)$.
\end{Corollary}

  \begin{Remark}\label{rem:nlp}
	Our focus here is on cases where~$\kappa(t)$ is explicitly known. However, the derived
	bounds are useful even when this is not the case. For example, suppose that~$\kappa(t)$ is not known, yet  
	the bounds 
	\begin{align*}
	 C_r&:=\max_{t\in [0,T]}|\kappa_r(t)|, \quad r=1,\ldots,n, \\
	C_v&:=\max_{t\in [0,T]}|u(t)|,
	\end{align*}
	are known. Then~\eqref{eq:nbh} implies that
	\begin{alignat*}{2}
	 |\gamma(\tau)-\kappa(\tau)|
  &&\leq \hspace{.4in}\frac{1}{\eta}\max_{z,v} \quad &|H(z,v)|\\
 &&\quad \textup{ subject to } |z_r|&\leq C_r,\quad r=1,\ldots,n\\
 &&|v|&\leq C_v.
\end{alignat*}
 This nonlinear optimization
program  is useful because   the feasible set is a box constraint and thus is a convex set.
	\end{Remark}

\begin{Example}  \label{exa:rfm2bnd}
We  again consider the RFM with $n=2$ and the periodic initiation rate~$u_0(t):=\lambda_0+u(t)$, with~$\lambda_0>1$ and~$u(t)=\sin(2\pi t/T)$. 
Again, let~$e$ be the unique equilibrium of the system 
when the initiation rate is~$\lambda_0 $ (see~\eqref{eq:e2exp}).  
Let~$\delta x:=x-e$. Then the linearized system is~$\dot{\delta x}=A \delta x+b u$, where
 \begin{align*}
  A=
  \begin{bmatrix}
    -\lambda_0-\lambda_1 (1-e_2)&\lambda_1 e_1\\
    \lambda_1(1-e_2)& -\lambda_1 e_1 - \lambda_2
  \end{bmatrix},\quad b=
                     \begin{bmatrix}
                       1-e_1\\0
                     \end{bmatrix}.
\end{align*}
Note that~$\mu_1(A)=\max\{-\lambda_0,-\lambda_2\}<0$, so, in particular,~$A$ is Hurwitz. 
  Thus, the approximating system is
\begin{align}
  \label{eq:3}
  \dot{y}=A(y-e)+ bu=:G(y,u),\quad u(t)=\sin(2\pi t/T).
\end{align}
The difference between the vector fields  evaluated 
 along a solution of the~$y$ system is
\begin{align*}
\nonumber &F(y, \sin(2\pi t/T))-  G(y, \sin(2\pi t/T))= \begin{bmatrix}
    \lambda_1(y_1-e_1)(y_2-e_2)-(y_1-e_1)\sin(2\pi t/T)\\-\lambda_1(y_1-e_1)(y_2-e_2)
  \end{bmatrix}.
\end{align*}
Let
$
  \hat{g}(s):=\begin{bmatrix}\hat{g}_1(s)\\\hat{g}_2(s)\end{bmatrix}=(sI-A)^{-1}b ,
$
and let  $\kappa(t):\R \to \R^2$ be the unique periodic trajectory of~\eqref{eq:3} defined for all $-\infty<t<\infty$.
Then
\begin{align*}
								\kappa(t) -e= \begin{bmatrix}
								|\hat g_1(j\omega)|\sin(  \omega t +\angle \hat g_1(j\omega) )\\
								|\hat g_2 (j\omega) |\sin( \omega t +\angle \hat g_2(j\omega)  )  
										\end{bmatrix},
\end{align*}
with~$\omega:=2\pi/T$.  By Remark~\ref{rem:nlp},  
\begin{alignat}{3}
\nonumber  |\gamma(t)-\kappa(t)|_1&\leq& \max_{z_1,z_2,v}\ &\frac{1}{\eta}     (|\lambda_1z_1z_2-z_1v|+|\lambda_1z_1z_2|)\\
\nonumber&&\text{subject to }|z_1|&\leq |\hat{g}_1(j\omega)|\\
\nonumber&&|z_2|&\leq |\hat{g}_2(j\omega)|\\
&&|v|&\leq 1 \nonumber \\
\label{eq:bnmfp}&=\frac{1}{\eta} \left( 
	2\lambda_1 |\hat{g}_1(j\omega )||\hat{g}_2(j\omega)|+|\hat{g}_1(j\omega)|\right)\hspace{-3in}
\end{alignat}
where~$\eta:=\min\{\lambda_0-1,\lambda_2\}$ as before. 
Note that the bound here depends on the frequency of the
 periodic excitation.
The more exact bound in~\eqref{eq:26} 
  can be computed numerically. 

For the parameters~$\lambda_1=1/2$, $\lambda_2=4$, and~$T=2$,
Figure~\ref{fig:err_period_ex2_phase} shows the equilibrium point when $\lambda_0=4$,  
the periodic trajectory for the case when the initiation rate is~$u_0(t)=4+\sin(2\pi t/T)$, 
and the  periodic trajectory  of the linearized system. 
Figure~\ref{fig:err_period_ex2} illustrates the bounds from Theorem~\ref{thm:linearized}.
It may be observed that these bounds provide a reasonable estimate of the error. 
\qedwhite
\end{Example}

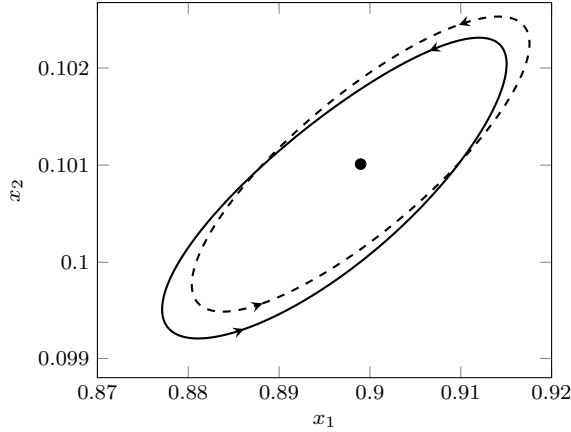
\begin{figure}[t]
{\footnotesize
  \begin{center}
\pgfplotsset{scaled y ticks=false}
\pgfplotsset{compat=1.12}
\begin{tikzpicture}
    \begin{axis}[width=3in,
   xmin=0.87, xmax=0.92, ymax=.1027,ymin=.0988,
    xlabel={$x_1$},
    ylabel={$x_2$},
    ytick={.0990,.100,...,.103},
scaled ticks = false, 
x tick label style={
        /pgf/number format/.cd,
        precision=2,
    },
y tick label style={
        /pgf/number format/.cd,
        precision=4,
    },
    legend style={at={(1.02,.98)},anchor=north west},
    ]
    \addplot[black, mark=*, only marks] coordinates {(0.8989794856,	0.1010205144)};
    \addplot[black, line width=.8pt] file{Figures/ex2_periodic_traj.txt} [arrow inside={}{0,.5}]; 
    \addplot[black, dashed, line width=.8pt] file{Figures/ex2_periodic_traj_approx.txt} [arrow inside={}{0,.5}];
\end{axis}
\end{tikzpicture}
\end{center}
}
  \caption{   RFM in Example~\ref{exa:rfm2bnd}.  The equilibrium~$e$
	for $\lambda_0= 4$  is marked by a dot. 
	The periodic trajectory~$\gamma(t)$  of the RFM  (solid line)
	and the   periodic trajectory~$\kappa(t)$  of the linearized system (dashed line)
	when~$u_0(t)=4+\sin(2\pi t/T)$.} 
\label{fig:err_period_ex2_phase}
\end{figure}

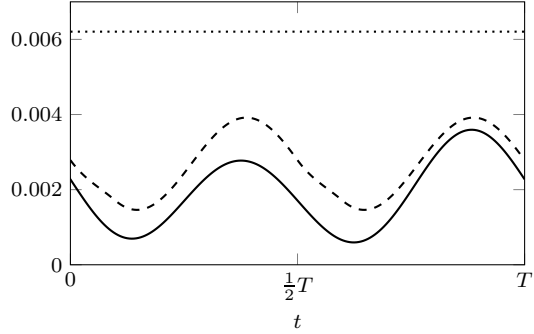
\begin{figure}[t]
{\footnotesize
  \begin{center}
\pgfplotsset{scaled y ticks=false}
\pgfplotsset{compat=newest}
\begin{tikzpicture}
    \begin{axis}[width=3in, height=2in,
   xmin=0, xmax=2, ymax=.007,ymin=0,
    xlabel={$t$},
    xtick={0,1,2},
    xticklabels={$0$,$\frac{1}{2}T$, $T$},
scaled ticks = false, 
    legend style={at={(1.02,.98)},anchor=north west},
y tick label style={
        /pgf/number format/.cd,
        precision=4,
        /pgf/number format/fixed,
    },
    ]
    \addplot[black, line width=.8pt] file{Figures/ex2_periodic_diff.txt};
    \addplot[black, dashed, line width=.8pt] file{Figures/ex2_periodic_bound.txt}; 
    \addplot[black, dotted, line width=.8pt] coordinates{(0,0.0062057604615785569) (2,0.0062057604615785569)};
\end{axis}
\end{tikzpicture}
\end{center}
}
\caption{
RFM in  Example~\ref{exa:rfm2bnd}. The error~$|\gamma(t)-\kappa(t)|_1$ (solid line) 
and the bounds~\eqref{eq:26} (dashed lines) and \eqref{eq:bnmfp} (dotted line). }
\label{fig:err_period_ex2}
\end{figure}

The bound~\eqref{eq:bnmfp} has some interesting implications. For example, if~$\hat g_1(j\omega)=0$ for some~$\omega$ then~\eqref{eq:bnmfp}  implies that~$\gamma(t)\equiv \kappa(t)$
 for a sinusoidal excitation with frequency~$\omega$.
 Similarly, if~$\lim_{\omega\to\infty}\hat g_1(j \omega) =0$ then~\eqref{eq:bnmfp} implies that
  for a high frequency sinusoidal forcing term,~$\gamma$ will approach~$\kappa$. Note that 
	the conclusions on~$\gamma$ here are based on \emph{properties of the~LTI system}.
	In the next section, we use
this idea to derive a theoretical result on the response of contractive systems to a sinusoidal input.  
\section{Contractive Systems as Low-Pass Filters}\label{sect:theory}

We consider a  contractive systems with an additive input and show that for  
  a high-frequency sinusoidal input, the periodic trajectory of the contractive system
	is very similar to that of a suitable~LTI system.
		For the sake of simplicity, we state this for the case of a scalar control.
\begin{Theorem}\label{thm:lpf}
Consider the system
\be\label{eq:FF}
\dot x=f(x)+bu
\ee
where
\[
u(t)=a \cos(\omega  t+\phi) .
\]
 Suppose that the trajectories of~\eqref{eq:FF} evolve on a compact 
and convex state space~$\Omega\subset\R^n$.
 Assume that for some
 vector norm~$|\cdot|:\R^n\to\R_+$, and     induced
 matrix measure~$\mu:\R^{n\times n }\to\R$,
 \begin{align*}
			\mu\left(	\frac{\partial f}{\partial x} \left(x \right ) \right)\leq -\eta<0   
 \end{align*}
for  all~$x \in \Omega$. Denote~$T:=2\pi/ \omega$, and 
let $\gamma(t)$ be the unique, attracting, $T$-periodic orbit of~\eqref{eq:FF} in~$\Omega$.

Suppose also that for the unforced dynamics,
i.e.~$\dot x=f(x )$ there exists a locally stable equilibrium
 point~$e \in \Omega$, and
without loss of generality, that~$e=0$. 
Let~$A:=\frac{\partial f}{\partial x} (0)$, and consider the approximating system 
  \begin{align} \label{eq:ling}
				\dot y  =Ay+ b  u:=G(y,u ).
  \end{align} 
	Let~$\hat g(s):=(sI-A)^{-1}b$ and 
let~$\kappa(t)$ be the unique $T$-periodic trajectory of~\eqref{eq:ling}, that is,
\be\label{eq:kpom}
\kappa_r(t)= a |\hat g_r(j\omega)|   \cos(\omega t+\phi+\angle \hat g_r(j \omega)),\quad r=1,\dots,n .
\ee
 Then   
\be\label{eq:boute}
\max_{t\in[0,T]} |  \gamma (t)-\kappa(t) |=  o(1/\omega).
\ee
\end{Theorem}

\begin{proof}
It follows from Corollary~\ref{coro:ftrans}
that for any~$\tau$, 
\begin{align} \label{eq:lkip}
  |\gamma(\tau)-\kappa(\tau)|&
  \leq \frac{1}{\eta} \max_{t\in[0,T]}\left |
	f   (\kappa(t))   - A \kappa(t)   \right | \nonumber  \\
	&\leq \frac{1}{\eta} \max_{t\in[0,T]} 
	o (|\kappa(t)|)        .
\end{align}
 Since~$\hat g(s)=\frac{\Adj(sI-A)  }{\det(sI-A)}  b $,  where~$\Adj$ denotes the adjugate,~\eqref{eq:kpom}
 implies that~$ | \kappa_r (t)|=O(1/ \omega)$ for all~$r$ and all~$t \in[0,T]$.
Combining this with~\eqref{eq:lkip}
completes the proof.
\end{proof}

The next two examples demonstrate
  Theorem~\ref{thm:lpf}.
\begin{Example}\label{exa:new_from_russo}
We consider a
 basic model for an externally driven  
 transcriptional module that is  ubiquitous in both
biology and synthetic biology (see, e.g.,~\cite{retro_2008,entrain2011}):
\begin{align} \label{eq:trans_module}
\dot x_1=& u-\delta x_1+k_1 x_2-k_2(e_T-x_2)x_1, \nonumber \\
\dot x_2=& - k_1 x_2 +k_2(e_T-x_2)x_1,
\end{align}
where~$\delta,k_1,k_2,e_T$ are strictly positive parameters.
Here~$x_1(t)$ is the concentration at time~$t$ of a transcriptional factor~$X$
that regulates  a downstream
transcriptional module by binding to a promoter
 with concentration~$e(t)$ yielding a protein-promoter complex~$Y$ 
with concentration~$x_2(t)$.
 The binding reaction   is reversible with  binding and
dissociation rates~$ k_2$ and $k_1$, respectively. The linear degradation rate of~$X$ is~$\delta$,
and as the promoter is not
subject to decay, its total concentration, $e_T$, is conserved,
so~$e(t)=e_T-x_2(t)$ for all~$t\geq 0$. The input~$u(t) $ might 
 represent for example the concentration of an enzyme or of a second
messenger that activates~$X$, so we assume that~$u(t)\geq 0$ for all~$t\geq 0$. 

Trajectories of~\eqref{eq:trans_module} evolve on~$[0,\infty)\times[0,e_T]$. 
For an input satisfying~$0 \leq u(t)\leq c$ for all~$t\geq 0$, the set
$\Omega  := [0, (c+k_1 e_T)/\delta] \times[0,e_T]$ is 
 a convex and compact   invariant  set.  

 Ref.~\cite{entrain2011}  has shown  that~\eqref{eq:trans_module} is
contractive with respect to  a certain  weighted   $L_1$ norm.
Indeed, the  Jacobian of~\eqref{eq:trans_module}
is
\begin{align*}
  J(x)=\begin{bmatrix}    -\delta-k_2(e_T-x_2)& k_1+k_2 x_1 \\
k_2(e_T-x_2) &-k_1-k_2 x_1\end{bmatrix},
\end{align*}
so for~$D:=\diag(d,1)$, with~$d>0$, 
\be\label{eq:djaco}
DJ(x)D^{-1}=\begin{bmatrix}    -\delta-k_2(e_T-x_2)& ( k_1+k_2 x_1) d \\
k_2(e_T-x_2)/d &-k_1-k_2 x_1\end{bmatrix}.
\ee
The off-diagonal terms here are non-negative, 
and this means that
		  for any~$d\in( \frac{k_2e_T}{k_2e_T+\delta}   ,1)$,
\begin{align*}
		\mu_{1,D}(J(x) )&\leq -\eta,\text{ for all } \begin{bmatrix} x_1&x_2 \end{bmatrix}'\in\Omega,
\end{align*}
where~$\eta:= \min\{   k_1(1-d)        , \delta+k_2 e_T  (1-d^{-1}) \}  >0  $.
Thus,~\eqref{eq:trans_module} is contractive with respect to the scaled norm~$|\cdot|_{1,D}$
with contraction rate~$\eta$.

Linearizing~\eqref{eq:trans_module} yields~$\dot y=G(y,u)=Ay+b u$,  with
\begin{align*} %
A:=\begin{bmatrix}
-\delta -k_2 e_T  &  k_1     \\
k_2 e_T & - k_1
\end{bmatrix},\quad b:=\begin{bmatrix}
1      \\
 0
\end{bmatrix}, 
\end{align*}
and
\begin{align*}
			\hat g(s)&=(sI-A)^{-1}b=\frac{1}{s^2+(  \delta+k_1+k_2 e_T)s+\delta k_1}\begin{bmatrix}
							s+k_1\\k_1
			\end{bmatrix}.  
\end{align*}
Since~$f(y)-Ay=    k_2 y_1 y_2\begin{bmatrix}1&-1\end{bmatrix} ' $, the bound~\eqref{eq:lkip}
yields
\begin{align}\label{eq:dbound}
  |D(\gamma(\tau)-\kappa(\tau))|_{1}&
  \leq \frac{k_2}{\eta} \max_{t\in[0,T]}|  \kappa_1(t)\kappa_2(t) \begin{bmatrix}1&-1\end{bmatrix} '   |_{1,D} \nonumber\\
	&\leq \frac{k_2(d+1)}{\eta} \max_{t\in[0,T]}|  \kappa_1(t)\kappa_2(t) |  . 
\end{align}
Note that for any   input in the form
\[
u(t)=\sum_{i=1}^p a_i\cos(\omega_i t),
\]
 the periodic trajectory~$\kappa(t)$
 is explicitly known 
and thus the bound~\eqref{eq:dbound} is explicit. 
Figure~\ref{fig:trans10} 
depicts  the trajectories of both
the contractive system~\eqref{eq:trans_module} and of the LTI system 
for the parameters~$k_1=1$, $k_2=5$, $\delta=1$, $e_T=2$, 
and the excitation~$u(t)=\cos(\omega t)$ 
for two different values of~$\omega$.\footnote{This control is not positive for all times, yet 
  for the initial conditions in the simulations the trajectory remains in a convex and compact region
in which  the off-diagonal terms in~\eqref{eq:djaco} are non-negative and contraction holds.}
It may be seen that for a larger value of~$\omega$ the difference between~$\gamma$ and~$\kappa$ decreases,
as anticipated by~\eqref{eq:boute}. 
\qedwhite
\end{Example}

\begin{figure}[t]
  \begin{center}
\pgfplotsset{compat=newest}
\begin{tikzpicture}
    \begin{axis}[width=.4\linewidth,
   xmin=-.15, xmax=.15, ymax=1.25,ymin=-1.25,
    xlabel={$x_1$},
    ylabel={$x_2$},
scaled ticks = false, 
    legend style={at={(1.02,.98)},anchor=north west},
    ]
    \addplot[black, dashed, line width=.8pt] file{Figures/ex4_lin_1.txt}  [arrow inside={}{0,.5}]; 
    \addplot[black, line width=.8pt] file{Figures/ex4_exact_1.txt}  [arrow inside={}{0,.5}]; 
\end{axis}
\end{tikzpicture}
\begin{tikzpicture}
    \begin{axis}[width=.4\linewidth,
   xmin=-.11, xmax=.11, ymax=.2,ymin=-.2,
    xlabel={$x_1$},
    ylabel={$x_2$},
    xtick={-.1,0,.1},
    legend style={at={(1.02,.98)},anchor=north west},
    ]
    \addplot[black, dashed, line width=.8pt] file{Figures/ex4_lin_5.txt}  [arrow inside={}{0,.5}]; 
    \addplot[black, line width=.8pt] file{Figures/ex4_exact_5.txt}  [arrow inside={}{.05,.55}]; 
\end{axis}
\end{tikzpicture}
  \caption{Trajectories~$\gamma$ (solid line)  and $\kappa$ (dashed line) 
	for the system in Example~\ref{exa:new_from_russo} for 
	$\omega=1$ (top) and~$\omega=5$ (bottom). Note the different scales in the figures.}  \label{fig:trans10}
  \end{center}
\end{figure}
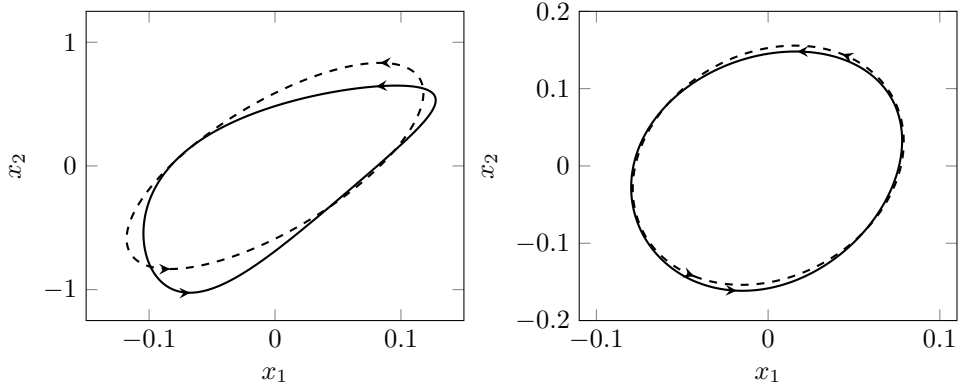

The next example demonstrates the result in Theorem~\ref{thm:lpf} using a nonlinear system for which the
frequency response has been computed explicitly in~\cite{pavlov2007}.
\begin{Example}\label{exa:frew}
Consider the system:
\begin{align}\label{eq:exnpmo}
							\dot x_1&=-x_1+x_2^2,\nonumber \\
							\dot x_2&=-x_2+u,
\end{align}
where the excitation is~$u(t)=a \sin(\omega t)$, with~$a,\omega >0$. It is clear that~$\Omega_2:=[-a,a]$ is an invariant set of~$x_2$. The Jacobian  of~\eqref{eq:exnpmo} is 
$J(x)=\begin{bmatrix}   -1&2x_2\\0&-1   \end{bmatrix}$. For any $c>0$ and $D:=\diag(1,c )$,  we have
$DJ(x)D^{-1}=\begin{bmatrix}   -1&2x_2/c\\0&-1   \end{bmatrix}$,
so~$\mu_1(DJ(x)D^{-1})\leq -1+\frac{2a}{c}$ for all~$x_2\in\Omega_2$. If $c>2a$, then this systems is contractive with respect
to the scaled norm~$|z|_{1,D}:=|Dz|_1$ with contraction rate
\be\label{eq:etfde}
\eta=1-\frac{2a}{c}.
\ee
 Note that, by taking~$c$ arbitrarily large, we may obtain a contraction rate arbitrarily close to~$ 1$. The periodic trajectory~$\gamma(t)$ can be computed explicitly as follows. First, it is clear 
that
\begin{align*}
\gamma_2(t)=\frac{a}{\sqrt{1+\omega^2} } \sin(\omega t-\tan^{-1}(\omega) ),  
\end{align*}
and substituting this in the first equation of~\eqref{eq:exnpmo} yields
\begin{align}
  \label{eq:29} \gamma_1(t)&=M\left[1+5\omega^2+4\omega^4+(5\omega^2-1)\cos(2\omega t) +2 \omega(\omega^2-2)\sin(2\omega t) \right ],
\end{align}
where~$M:=\frac{a^2}{2(1+\omega^2)^2(1+4\omega^2)}$.

Note that the  unforced dynamics admits an equilibrium~$e=0$. The approximating system is~$\dot y=G(y,u)=-y+b u$, with~$b:=\begin{bmatrix}0&1\end{bmatrix}'$ and~$\hat g(s)=(s+1)^{-1} b$. Thus, $\kappa(t)=\begin{bmatrix} 0 & \gamma_2(t) \end{bmatrix}'$,
so~$\gamma(t)-\kappa(t)  =    \begin{bmatrix} \gamma_1(t) & 0 \end{bmatrix}' $,
and~$|\gamma(t)-\kappa(t) |_{1,D}=|\gamma_1(t)|$.
To apply Corollary~\ref{coro:ftrans} note that~$H(z,v) :=F(z,v)-G(z,v)=\begin{bmatrix}z_2^2&0\end{bmatrix}'$, 
so applying the bound~\eqref{eq:nbh}   gives
\begin{align}
\max_{t\in [0,T]}|\gamma(t)-\kappa(t) |_{1,D}&=\max_{t\in[0,T]}|\gamma_1(t)| \nonumber \\
&\leq \eta^{-1}\max_{t\in[0,T]}|\kappa^2_2(t)|  \nonumber \\
  \label{eq:42-3}&=\eta^{-1} \frac{a^2}{   1+\omega^2 }
\end{align}
where~$\eta$ is given in~\eqref{eq:etfde} with~$c>2a$,
and~$T=2\pi/\omega$. Taking $c\to\infty$ gives the explicit bound~$ \max_{t\in [0,T]}|\gamma(t)-\kappa(t) |_{1,D} \leq a^2/(1+\omega^2)$.
In fact, it follows from \eqref{eq:29}, after some calculation, that 
\begin{align}
  \label{eq:23}\max_{t\in[0,T]}|\gamma_1(t)|&=\frac{a^2(1+\sqrt{4\omega^2+1})}{2(1+\omega^2) \sqrt{4\omega^2+1}}.
\end{align}

Fig.~\ref{fig:gakap} depicts the exact difference~$\max_{t\in [0,T]}|\gamma(t)-\kappa(t) |_{1,D}$ given in \eqref{eq:23} and the bound $a^2/(1+\omega^2)$ implied by \eqref{eq:42-3}, as a function of~$\omega$
for~$a=1$. Theorem \ref{thm:lpf} guarantees $\max_{t\in[0,T]} |  \gamma (t)-\kappa(t) |=  o(1/\omega)$, as seen in the figure.
\qedwhite
\end{Example}
\begin{figure}
  \centering
  \begin{tikzpicture}
\pgfplotsset{compat=newest}
    \begin{axis}[width=.45\linewidth,
    xmode=log,
    ymode=log,
   xmin=.01, xmax=100, ymax=2,ymin=.0005,
    xlabel={$\omega$},
    ylabel={\footnotesize $ \max_{t}|\gamma(t)-\kappa(t) |$},
x label style={at={(axis description cs:.5,-.35)},anchor=north},
    grid=major,
legend cell align=left,
    yscale=.8,
    name=mag,
    legend style={at={(.99,1)},anchor=north east},
    ]
    \addplot[dashed, line width=1pt] file{Figures/ex6bnd.txt};
    \addplot[black, line width=1pt] file{Figures/ex6exact.txt};
    \addlegendentry{\footnotesize Bound};
    \addlegendentry{\footnotesize Exact};
\end{axis}
\end{tikzpicture}
  \caption{Maximum distance between the exact periodic trajectory~$\gamma(t)$ and the approximate periodic trajectory~$\kappa(t)$ as a function of the excitation frequency~$\omega$ for Example~\ref{exa:frew} with~$a=1$. The solid line is the exact difference, and the dashed line is the bound on the difference determined by Corollary~\ref{coro:ftrans}.}
  \label{fig:gakap}
\end{figure}
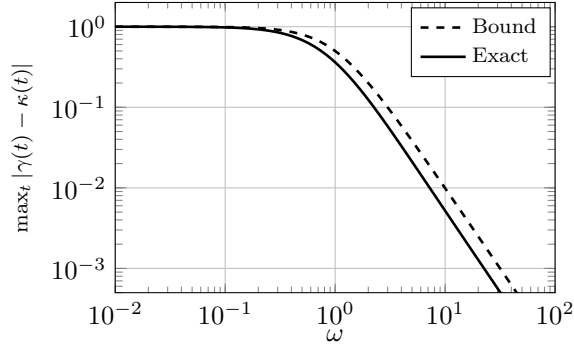

\section{Discussion}

Contractive systems entrain to   periodic excitations. 
 Analyzing the   corresponding     periodic solution of the contractive system and its dependence on
  various parameters is an important   theoretical question with many potential applications.
We developed approximation schemes for this periodic solution using LTI systems 
 and, using the~ISS property
of contractive systems, provided bounds on the approximation error. An important advantage of 
these bounds is that in some cases they can be computed explicitly.  This also led  
 to a  new theoretical result  on the behavior of contractive systems for
a high frequency   excitation.

More generally, it is well-known that contractive systems whose solutions
evolve on a compact state-space have a well-defined frequency response~\cite{pavlov2007,Ruffer2013277}.
For the contractive system $\dot{x}=F(x,u)$, with~$x\in \R^n$ and~$u \in\R$, this means 
 that 
there exists a continuous function~$\alpha:\mathbb{R}^3\to \mathbb{R}^n$ such that
the following property holds. 
For the sinusoidal input $u(t)=a\sin(\omega t)$, with frequency $\omega:=2\pi/T$ and amplitude $a\geq 0 $,
the solution of the contractive system converges to a periodic solution~$\gamma_{a\omega}$
satisfying 
\[
  \gamma_{a\omega}(t)=\alpha(a\sin(\omega t),a \cos(\omega t), \omega)
\]
(see \cite[Theorem 3]{pavlov2007}). The function $\alpha(v_1,v_2,w)$ is called the
 \emph{state frequency response}. For the special case of a linear system, i.e. $F(x,u)=Ax+bu$
the state frequency response is known explicitly:
\[
  {\alpha}(v_1,v_2,\omega)=\Pi(\omega)\begin{bmatrix}v_1&v_2\end{bmatrix}'
\]
with   $\Pi(\omega):=\begin{bmatrix}\text{Re}(\hat{g}(j\omega))&\text{Im}(\hat{g}(j\omega))\end{bmatrix}\in\mathbb{R}^{n\times 2}$,
and $\hat{g}(s):=(sI-A)^{-1}b$. 
That is, for linear systems, the state frequency response recovers the standard notion of frequency response.

Of course, for nonlinear systems it is typically  not possible to compute
the frequency response analytically.  Our results may be interpreted in this context as follows. 
Considering Theorem \ref{thm:averexci}, we have that $\bar{u}=\frac{1}{T}\int_0^T u(t) dt=0$ for any $a$ and $\omega$  and that~$\dot y= F(y,0)$ admits an equilibrium point~$e$. Thus,
$e=\alpha(0,0,\omega)$ (where $e$ is in fact independent of $\omega$), and~\eqref{eq:20}  may be 
interpreted as providing bounds on 
\[ 
|\alpha(a\sin(\omega t),a\cos(\omega t),\omega)-e| .  
\]

On the other-hand, the results in 
Theorem~\ref{thm:linearized}
may be interpreted as bounds on the difference 
\[
  \left|\alpha(a\sin(\omega t),a\cos(\omega t),\omega)-\bar{\alpha}(a\sin(\omega t),
	a\cos(\omega t),\omega)\right|,  
\]
where $\bar{\alpha}$ is the state frequency response of the linearized 
system~$\dot{y}=Ay+Bu$, 
with~$A=\frac{\partial F}{\partial x}(0,0)$ and~$B=\frac{\partial F}{\partial u}(0,0)$.

An interesting topic for further research is deriving more
  theoretical results using the explicit bounds described here. 
Other  possible topics include the design of an excitation signal that yields a pre-specified 
 periodic trajectory for a contractive system. This issue is  important for example
in synthetic biology, where
an important goal is to design programmable biochemical  oscillators (see e.g.,~\cite{elo_clock,fung_ocs,str_osc,weitz_osc}).
Another possible research topic is  the extension of the  results presented here
to more general classes of dynamical systems (see, e.g.,~\cite{nata_george} for a special class of infinite dimensional systems
 that admit a frequency response).

\section*{Acknowledgments} We are grateful to Eduardo  D. Sontag for reading an earlier version of this paper and providing us with many useful comments.

 \bibliographystyle{IEEEtran}

\bibliography{sam_bib}

\begin{thebibliography}{10}
\providecommand{\url}[1]{#1}
\csname url@samestyle\endcsname
\providecommand{\newblock}{\relax}
\providecommand{\bibinfo}[2]{#2}
\providecommand{\BIBentrySTDinterwordspacing}{\spaceskip=0pt\relax}
\providecommand{\BIBentryALTinterwordstretchfactor}{4}
\providecommand{\BIBentryALTinterwordspacing}{\spaceskip=\fontdimen2\font plus
\BIBentryALTinterwordstretchfactor\fontdimen3\font minus
  \fontdimen4\font\relax}
\providecommand{\BIBforeignlanguage}[2]{{%
\expandafter\ifx\csname l@#1\endcsname\relax
\typeout{** WARNING: IEEEtran.bst: No hyphenation pattern has been}%
\typeout{** loaded for the language `#1'. Using the pattern for}%
\typeout{** the default language instead.}%
\else
\language=\csname l@#1\endcsname
\fi
#2}}
\providecommand{\BIBdecl}{\relax}
\BIBdecl

\bibitem{LOHMILLER1998683}
W.~Lohmiller and J.-J.~E. Slotine, ``{On contraction analysis for non-linear
  systems},'' \emph{Automatica}, vol.~34, pp. 683--696, 1998.

\bibitem{sontag_cotraction_tutorial}
Z.~Aminzare and E.~D. Sontag, ``Contraction methods for nonlinear systems: A
  brief introduction and some open problems,'' in \emph{{Proc.\ 53rd IEEE Conf.
  on Decision and Control}}, Los Angeles, CA, 2014, pp. 3835--3847.

\bibitem{entrain2011}
G.~Russo, M.~di~Bernardo, and E.~D. Sontag, ``Global entrainment of
  transcriptional systems to periodic inputs,'' \emph{{PLOS Computational
  Biology}}, vol.~6, p. e1000739, 2010.

\bibitem{Russo2011_book_chap}
G.~Russo, M.~di~Bernardo, and J.~J. Slotine, ``Contraction theory for systems
  biology,'' in \emph{Design and Analysis of Biomolecular Circuits: Engineering
  Approaches to Systems and Synthetic Biology}, H.~Koeppl, G.~Setti,
  M.~di~Bernardo, and D.~Densmore, Eds.\hskip 1em plus 0.5em minus 0.4em\relax
  New York, NY: Springer, 2011, pp. 93--114.

\bibitem{hier_cont}
G.~Russo, M.~{di Bernardo}, and E.~D. Sontag, ``A contraction approach to the
  hierarchical analysis and design of networked systems,'' \emph{{IEEE Trans.\
  Automat.\ Control}}, vol.~58, pp. 1328--1331, 2013.

\bibitem{Arcak20111219}
M.~Arcak, ``Certifying spatially uniform behavior in reaction-diffusion {PDE}
  and compartmental {ODE} systems,'' \emph{Automatica}, vol.~47, no.~6, pp.
  1219--1229, 2011.

\bibitem{cont_slotine_graph}
G.~Russo, M.~di~Bernardo, and J.~J.~E. Slotine, ``A graphical approach to prove
  contraction of nonlinear circuits and systems,'' \emph{IEEE Transactions on
  Circuits and Systems I: Regular Papers}, vol.~58, no.~2, pp. 336--348, 2011.

\bibitem{coogan2015compartmental}
S.~Coogan and M.~Arcak, ``A compartmental model for traffic networks and its
  dynamical behavior,'' \emph{{IEEE Trans.\ Automat.\ Control}}, vol.~60,
  no.~10, pp. 2698--2703, 2015.

\bibitem{Desoer_cont}
C.~Desoer and H.~Haneda, ``The measure of a matrix as a tool to analyze
  computer algorithms for circuit analysis,'' \emph{IEEE Trans. Circuit
  Theory}, vol.~19, pp. 480--486, 1972.

\bibitem{sontag2008}
E.~D. Sontag, ``Input to state stability: Basic concepts and results,'' in
  \emph{Nonlinear and Optimal Control Theory}, P.~Nistri and G.~Stefani,
  Eds.\hskip 1em plus 0.5em minus 0.4em\relax Berlin, Heidelberg: Springer,
  2008, pp. 163--220.

\bibitem{Arcak2015}
J.~Maidens and M.~Arcak, ``Reachability analysis of nonlinear systems using
  matrix measures,'' \emph{{IEEE Trans.\ Automat.\ Control}}, vol.~60, no.~1,
  pp. 265--270, 2015.

\bibitem{con_with_input}
E.~D. Sontag, ``Contractive systems with inputs,'' in \emph{Perspectives in
  Mathematical System Theory, Control, and Signal Processing}, J.~Willems,
  S.~Hara, Y.~Ohta, and H.~Fujioka, Eds.\hskip 1em plus 0.5em minus 0.4em\relax
  Berlin Heidelberg: Springer-Verlag, 2010, pp. 217--228.

\bibitem{lars16}
\BIBentryALTinterwordspacing
M.~Botner, Y.~Zarai, M.~Margaliot, and L.~Gr\"une, ``On approximating
  contractive systems,'' \emph{{IEEE Trans.\ Automat.\ Control}}, 2017, {To
  appear}. [Online]. Available:
  \url{http://ieeexplore.ieee.org/document/7814289/}
\BIBentrySTDinterwordspacing

\bibitem{soderling_survey}
G.~Soderlind, ``The logarithmic norm. {H}istory and modern theory,'' \emph{BIT
  Numerical Mathematics}, vol.~46, pp. 631--652, 2006.

\bibitem{cont_anc}
J.~Jouffroy, ``Some ancestors of contraction analysis,'' in \emph{{Proc.\ 44th
  IEEE Conf. on Decision and Control}}, Seville, Spain, 2005, pp. 5450--5455.

\bibitem{contra_sep}
F.~Forni and R.~Sepulchre, ``A differential {L}yapunov framework for
  contraction analysis,'' \emph{{IEEE Trans.\ Automat.\ Control}}, vol.~59,
  no.~3, pp. 614--628, 2014.

\bibitem{vid}
M.~Vidyasagar, \emph{{Nonlinear Systems Analysis}}.\hskip 1em plus 0.5em minus
  0.4em\relax Englewood Cliffs, NJ: Prentice Hall, 1978.

\bibitem{sandberg78}
I.~W. Sandberg, ``On the mathematical foundations of compartmental analysis in
  biology, medicine, and ecology,'' \emph{IEEE Trans. Circuits and Systems},
  vol.~25, no.~5, pp. 273--279, 1978.

\bibitem{Coogan:2016kx}
S.~Coogan, ``Separability of {L}yapunov functions for contractive monotone
  systems,'' in \emph{{Proc.\ 55th IEEE Conf. on Decision and Control}}, Las
  Vegas, NV, 2016, pp. 2184--2189.

\bibitem{RFM_entrain}
M.~Margaliot, E.~D. Sontag, and T.~Tuller, ``{Entrainment to periodic
  initiation and transition rates in a computational model for gene
  translation},'' \emph{PLOS ONE}, vol.~9, no.~5, p. e96039, 2014.

\bibitem{strom_log_norms}
T.~Strom, ``On logarithmic norms,'' \emph{SIAM J . Numerical Analysis},
  vol.~12, pp. 741--753, 1975.

\bibitem{reuveni}
S.~Reuveni, I.~Meilijson, M.~Kupiec, E.~Ruppin, and T.~Tuller, ``{Genome-scale
  analysis of translation elongation with a ribosome flow model},'' \emph{{PLOS
  Computational Biology}}, vol.~7, p. e1002127, 2011.

\bibitem{cast_book}
M.~Margaliot, E.~D. Sontag, and T.~Tuller, ``Checkable conditions for
  contraction after small transients in time and amplitude,'' in \emph{Feedback
  Stabilization of Controlled Dynamical Systems-In Honor of Laurent Praly},
  ser. Lecture Notes in Control and Information Sciences, N.~Petit, Ed.\hskip
  1em plus 0.5em minus 0.4em\relax Springer-Verlag, 2017, vol. 466.

\bibitem{3gen_cont_automatica}
------, ``Contraction after small transients,'' \emph{Automatica}, vol.~67, pp.
  178--184, 2016.

\bibitem{zarai_infi}
Y.~Zarai, M.~Margaliot, and T.~Tuller, ``Explicit expression for the
  steady-state translation rate in the infinite-dimensional homogeneous
  ribosome flow model,'' \emph{{IEEE/ACM Trans. Computational Biology and
  Bioinformatics}}, vol.~10, pp. 1322--1328, 2013.

\bibitem{RFM_feedback}
{Margaliot, M. and Tuller, T.}, ``{Ribosome flow model with positive
  feedback},'' \emph{J. Royal Society Interface}, vol.~10, p. 20130267, 2013.

\bibitem{RFM_stability}
M.~Margaliot and T.~Tuller, ``{Stability analysis of the ribosome flow
  model},'' \emph{{IEEE/ACM Trans. Computational Biology and Bioinformatics}},
  vol.~9, pp. 1545--1552, 2012.

\bibitem{RFMR}
A.~Raveh, Y.~Zarai, M.~Margaliot, and T.~Tuller, ``Ribosome flow model on a
  ring,'' \emph{{IEEE/ACM Trans. Computational Biology and Bioinformatics}},
  vol.~12, no.~6, pp. 1429--1439, 2015.

\bibitem{RFM_model_compete_J}
A.~Raveh, M.~Margaliot, E.~D. Sontag, and T.~Tuller, ``A model for competition
  for ribosomes in the cell,'' \emph{J. Royal Society Interface}, vol.~13, no.
  116, 2016.

\bibitem{RFM_sense}
G.~Poker, M.~Margaliot, and T.~Tuller, ``Sensitivity of {mRNA} translation,''
  \emph{Sci. Rep.}, vol.~5, p. 12795, 2015.

\bibitem{rfm_max_density}
Y.~Zarai, M.~Margaliot, and T.~Tuller, ``On the ribosomal density that
  maximizes protein translation rate,'' \emph{PLOS ONE}, vol.~11, no.~11, pp.
  1--26, 2016.

\bibitem{rfm_opt_down}
------, ``Optimal down regulation of {mRNA} translation,'' \emph{Sci. Rep.},
  vol.~7, no. 41243, 2017.

\bibitem{retro_2008}
D.~{Del Vecchio}, A.~J. Ninfa, and E.~D. Sontag, ``Modular cell biology:
  Retroactivity and insulation,'' \emph{Molecular Systems Biology}, vol.~4,
  no.~1, p. 161, 2008.

\bibitem{pavlov2007}
A.~Pavlov, N.~van~de Wouw, and H.~Nijmeijer, ``Frequency response functions for
  nonlinear convergent systems,'' \emph{{IEEE Trans.\ Automat.\ Control}},
  vol.~52, no.~6, pp. 1159--1165, 2007.

\bibitem{Ruffer2013277}
B.~S. Ruffer, N.~van~de Wouw, and M.~Mueller, ``Convergent systems vs.
  incremental stability,'' \emph{{Systems Control Lett.}}, vol.~62, no.~3, pp.
  277--285, 2013.

\bibitem{elo_clock}
M.~B. Elowitz and S.~Leibler, ``A synthetic oscillatory network of
  transcriptional regulators,'' \emph{Nature}, vol. 403, pp. 335--338, 2000.

\bibitem{fung_ocs}
E.~Fung, W.~W. Wong, J.~K. Suen, T.~Bulter, S.-g. Lee, and J.~C. Liao, ``A
  synthetic gene-metabolic oscillator,'' \emph{Nature}, vol. 435, pp. 118--122,
  2005.

\bibitem{str_osc}
J.~Stricker, S.~Cookson, M.~R. Bennett, W.~H. Mather, L.~S. Tsimring, and
  J.~Hasty, ``A fast, robust and tunable synthetic gene oscillator,''
  \emph{Nature}, vol. 456, pp. 516--539, 2008.

\bibitem{weitz_osc}
M.~Weitz, J.~Kim, K.~Kapsner, E.~Winfree, E.~Franco, and F.~C. Simmel,
  ``Diversity in the dynamical behaviour of a compartmentalized programmable
  biochemical oscillator,'' \emph{Nature Chemistry}, vol.~6, pp. 295--302,
  2014.

\bibitem{nata_george}
V.~Natarajan and G.~Weiss, ``Behavior of a stable nonlinear
  infinite-dimensional system under the influence of a nonlinear exosystem,''
  in \emph{Proc. 1st {IFAC} Workshop on Control of Systems Governed by Partial
  Differential Equations}, Paris, France, 2013, pp. 155--160.

\end{thebibliography}
\end{document}